\documentclass[conference]{IEEEtran}
\IEEEoverridecommandlockouts
\usepackage[left=1.5cm,right=1.6cm,top=1.9cm,bottom=2.3cm]{geometry}
\usepackage[final]{graphicx}
\usepackage{multicol}
\graphicspath{{Images//}}

\usepackage{algorithmic}
\usepackage{algorithm}
\usepackage{psfrag}
\usepackage{epstopdf}
\usepackage{amsfonts}
\usepackage{mathrsfs}
\usepackage{pifont}
\usepackage{verbatim}
\usepackage{upgreek}
\usepackage{color}
\usepackage[usenames,dvipsnames]{xcolor}
\usepackage{mdwmath}
\usepackage{mdwtab}
\usepackage{amssymb,amsmath}
\usepackage{amsmath}
\usepackage{amsthm}
\usepackage{epstopdf}
\usepackage{cite}
\usepackage{tikz}
\usepackage{scalefnt}
\usepackage{subfigure}
\usepackage{MnSymbol}
\usepackage{setspace}
\usepackage{mathtools}
\usepackage{psfrag}
\usepackage{booktabs}
\usepackage{siunitx}

\usepackage[font=scriptsize]{caption}

\UseRawInputEncoding

\label{newcommand}

\label{English Chars}
\newcommand{\bA}{\textbf{A}}
\newcommand{\ba}{\textbf{a}}

\newcommand{\bL}{\textbf{L}}

\newcommand{\bP}{\textbf{P}}

\newcommand{\bt}{\textbf{t}}

\label{Roman Chars}

\label{Theorems}
\newtheorem{theorem}{Theorem}

\theoremstyle{remark}

\def\BibTeX{{\rm B\kern-.05em{\sc i\kern-.025em b}\kern-.08em
    T\kern-.1667em\lower.7ex\hbox{E}\kern-.125emX}}

\begin{document}

\label{title}
\title{Near-Optimal LOS and Orientation Aware\\ Intelligent Reflecting Surface Placement}

\author{
	\IEEEauthorblockN{Ehsan Tohidi\IEEEauthorrefmark{1}\IEEEauthorrefmark{2}, Sven Haesloop\IEEEauthorrefmark{1}, Lars Thiele\IEEEauthorrefmark{1}, and Slawomir Sta\'nczak\IEEEauthorrefmark{1}\IEEEauthorrefmark{2}
		}
	\IEEEauthorblockA{\IEEEauthorrefmark{1}%
		Fraunhofer HHI, Berlin, 
\IEEEauthorrefmark{2} Technische Universität Berlin\\
Email: \{ehsan.tohidi, sven.haesloop, lars.thiele, slawomir.stanczak\}@hhi.fraunhofer.de}\,\\
\thanks{The authors acknowledge the financial support by the Federal Ministry of Education and Research of Germany in the programme of ``Souver\"an. Digital. Vernetzt.'' Joint project 6G-RIC, project identification number: 16KISK020K and 16KISK030.}
}

\maketitle

\begin{abstract}
Due to their passive nature and thus low energy consumption, intelligent reflecting surfaces (IRSs) have shown promise as means of extending coverage as a proxy for connection reliability. 
The relative locations of the base station (BS), IRS, and user equipment (UE) determine the extent of the coverage that IRS provides which demonstrates the importance of IRS placement problem. More specifically, locations, which determine whether BS-IRS and IRS-UE line of sight (LOS) links exist, and surface orientation, which determines whether the BS and UE are within the field of view (FoV) of the surface, play crucial roles in the quality of provided coverage.
Moreover, another challenge is high computational complexity, since IRS placement problem is a combinatorial optimization, and is NP-hard. Identifying the orientation of the surface and LOS channel as two crucial factors, we propose an efficient IRS placement algorithm that takes these two characteristics into account in order to maximize the network coverage. We prove the submodularity of the objective function which establishes near-optimal performance bounds for the algorithm. Simulation results demonstrate the performance of the proposed algorithm in a real environment.
\end{abstract}


\section{Introduction} 
Research on the sixth generation (6G) of wireless networks is being intensively pursued in order to find solutions to enable future applications, on the one hand, with high quality of service (QoS) requirements \cite{wu2019intelligent, pan2021reconfigurable, wu2021intelligent}, while on the other hand, reducing the network energy consumption and hardware costs. Intelligent reflecting surfaces (IRSs) are an emerging technology for future 6G wireless communication networks, which can make a decisive contribution to fulfilling these requirements \cite{wu2019intelligent, pan2021reconfigurable, wu2021intelligent}.
In contrast to densifying the network with more base stations (BS) or using other active nodes such as decode-and-forward (DF) or amplify-and-forward (AF) relays that process and amplify the signal, IRSs are passive in nature, i.e., they reflect the signal without amplification or any other complex operations such as encoding and decoding.
Thus, IRS has the potential for reducing hardware and energy costs compared to their active counterparts \cite{yang2021energy}, which enables a denser deployment in wireless networks.

Developing IRS-assisted wireless communication networks raises the problem of IRS placement, a part of the so-called network planning problem. In order to realize IRS placement, we need a model that captures the main characteristics of IRS, while still being mathematically tractable given the large scale of network planning problems. The question of ``what level of abstraction'' is then raised in order to maintain the right balance between accuracy and complexity.

Migration to higher frequency bands such as millimeter wave (mmWave) or the sub-terahertz (SubTHz) to achieve the ambitious requirements in 6G wireless networks on the one hand, and the product-distance path loss model of IRS on the other hand, imply the need for large surfaces with many elements, which can achieve passive beamforming gain through joint reflection \cite{wu2019intelligent}. To compensate for the severe path loss, besides the required large surface, the existence of the LOS path is indispensable, which is equivalent to the LOS link between BS-IRS and IRS-UE. In addition to improving the channel gain, the LOS link also simplifies the control and operation of an IRS, which is a typical assumption in the literature
\cite{9690635,9838993,9827797}. This issue, moreover, explains a prospective use case of IRSs, namely coverage extension, by providing virtual LOS links through the IRSs.
Furthermore, another important feature that to the best knowledge of the authors is largely neglected in the literature is the orientation of the surface. More specifically, the orientation of the surface with respect to the BS and UE. It is known that the further we move away from the surface broadside (i.e., toward the end fire), the higher the effects of impairments such as beam squint and gain loss become. Therefore, in a realistic model, we need to limit the field of view (FoV) of the surface, which is determined, for instance, based on the power budget and ensures that the BS and considered area of interest (i.e., technically where UEs are supposed to be) are within this FoV.

\subsection{Related Works}
The IRS placement problem has recently become of high interest, e.g., \cite{wu2021intelligent,ghatak2021placement,ntontin2020reconfigurable,lu2021aerial,hashida2020intelligent,9712623,issa2021using,9903366,zeng2020reconfigurable}.
Considering a fixed scenario, it is shown in \cite{wu2021intelligent} that the minimum product-distance path loss model is achieved by placing the IRS either close to the BS or close to the UE. However, the existence of the LOS path is not considered, and particularly in this case, the probability that both the direct link between BS and UE and the IRS-aided link being blocked simultaneously will increase.

The authors in \cite{ntontin2020reconfigurable} propose an IRS placement approach for highly-directional mmWave links considering transmission beam footprint at IRS plane as well as IRS size to maximize end-to-end signal-to-noise ratio (SNR). 
A joint placement and passive beamforming of aerial intelligent reconfigurable surfaces (AIRS) attached to a flying object such as an unmanned aerial vehicle (UAV) is proposed in \cite{lu2021aerial}. Moreover, in \cite{hashida2020intelligent}, an IRS placement optimization problem aiming at mitigation of inter-cell interference in air-ground communication networks is formulated, where a terrestrial IRS reflects the radio signal to an airborne receiver. Furthermore, in \cite{9712623}, an adaptive differential evolution algorithm is developed to optimize the number, locations, and phase shift coefficients of IRSs.
In \cite{issa2021using}, a method for the placement of multiple IRSs in indoor scenarios is developed that maximizes the achievable rate in all rooms. For this purpose, an indoor channel model is used, which takes obstacles into consideration through corresponding values of the Rician K-factors.
The authors of \cite{ghatak2021placement} incorporate objects to model blockage for optimal placement of two RISs in mmWave scenarios for received signal strength indicator (RSSI)-based localization of users.

\cite{9903366} is among the rare works that besides the location, investigates the effect of the orientation of an IRS on extended coverage. The authors confirm that the closer the angle of a user to the broadside of the surface is, the higher the achieved gain becomes. However, only one IRS and one single user is considered, only orientation towards the user is taken into account, and no algorithm for orientation and location optimization is proposed. As another IRS placement work considering orientation, authors in \cite{zeng2020reconfigurable} proposed an algorithm to obtain the optimal placement and orientation of an IRS. Since, in \cite{zeng2020reconfigurable}, only one IRS is assumed, and also, a free space is considered for the environment (i.e., no obstacles), the optimal orientation is always reported as the one with a broadside towards the BS.

\subsection{Contributions and Novelties}
This paper takes one step further in the direction of IRS placement as a component of IRS-assisted wireless communication networks with a  focus on coverage extension by establishing virtual LOS links between BSs and UEs. Taking into account the existing tradeoff between accuracy and complexity, we propose an improved model for IRS placement considering both the LOS path and orientation of the surface in order to tackle the IRS placement problem. 
The proposed model accepts a realistic environment in order to include obstacles and determine potential areas for IRS placement. Since IRS placement is a combinatorial discrete optimization problem and is NP-hard \cite{9712623}, we propose an efficient algorithm with near-optimal performance guarantees based on submodular optimization to tackle the placement problem, which is done for the first time to the best of the authors' knowledge. The contributions of this paper can be summarized as
follows:

\begin{itemize}
    \item We consider a more realistic environment including obstacles to determine the availability of LOS links.
    \item We enhance the system model by considering placement of multiple IRSs in addition to including two main factors, namely LOS and orientation, in the problem formulation
    \item We prove the submodularity of the optimization problem and propose an efficient placement algorithm with near-optimal performance guarantees.
\end{itemize}

\subsection{Outline and Notations}
The rest of the paper is organized as follows. In Section II, the problem is stated. Section III provides the algorithm and near-optimal guarantees. Simulation results are reported in
Section IV. Finally, Section V concludes the paper. We adopt the notation of using boldface lower (upper) case for vectors $\ba$ (matrices $\bA$).
$\mathbb{R}^{N\times M}$ is the set of $N \times M$ real matrices. Moreover, $|\mathcal{A}|$ is the cardinality of the set $\mathcal{A}$. Finally, $O(\cdot)$ is a mathematical notation that describes the order of a function or in other words, the computational complexity of a function.
\section{Problem Statement}

We consider a multi-IRS-assisted wireless communication system in a scenario as shown in Fig. \ref{fig.SystemModel}, which comprises a BS located at $\bP_B=[B_x,B_y,B_z]$, an IRS deployment area $\mathcal{P}=\{\bP_1,\ldots,\bP_M\}$ with $M$ IRS candidate locations, and an area of interest $\mathcal{L}=\{\bL_1,\ldots,\bL_N\}$ with $N$ potential locations for UEs, to be covered either by the BS or an IRS. 
We assume a LOS link to the BS location, $\bP_B$, for all $\bP\in\mathcal{P}$\footnote{If LOS to BS is not considered in $\mathcal{P}$, points with no LOS can be removed as a preprocessing step.}. 
Moreover, regarding the orientation of the surface,
we define the valid FoV of the surface as $[-\theta_{\text{max}},+\theta_{\text{max}}]$ with $\theta_{\text{max}}$ being the maximum acceptable deviation from the broadside 
(for both incident and reflecting angles), that satisfies the given system constraints (e.g., requirements on the maximum beam squint effect and gain loss) \footnote{Here, we only limit FoV on azimuth axis 
due to two main reasons: $(i)$ steered elevation angle is usually less of problem with the typical covering distances and installation heights, and $(ii)$ simplifying the notation of the problem. However, including elevation FoV constraint is straightforward, i.e., by adding a new dimension (elevation angle) to the feasible space of IRS placement.}. For each candidate location of the IRS, the range of valid orientations of the surface is determined in a way to ensure the steering angle towards the BS is in the range of valid FoVs. Fig. \ref{fig.ValidRange} illustrates an example of a valid range of azimuth angles for an arbitrary candidate position of IRS and the given BS location. The blue surface (line in the figure) is the one with broadside in the direction of the BS, and the red and green ones are the ones that reach $\pm\theta_{\text{max}}$ in steering angle toward the BS. Therefore, the valid range of orientations of the surface that provide a steering angle within the valid FoV are all orientations between the orientation of red and green surfaces. Similarly, for each candidate location of the IRS and for a given UE position, the range of valid orientations of the surface is calculated. This procedure should be repeated for the whole UE positions that are supposed to be covered by the considered IRS location (i.e., a subset of $\mathcal{L}$). As a result, for each IRS candidate location, given the BS location, and the set of UE positions that are supposed to be covered by the considered IRS, the final valid orientations of the surface are calculated by the intersection of all individually determined valid orientations of the surface, i.e., given by BS and the subset UE positions.


\begin{figure}[h]
     \centering
     \subfigure[]{\includegraphics[width=0.23\textwidth]{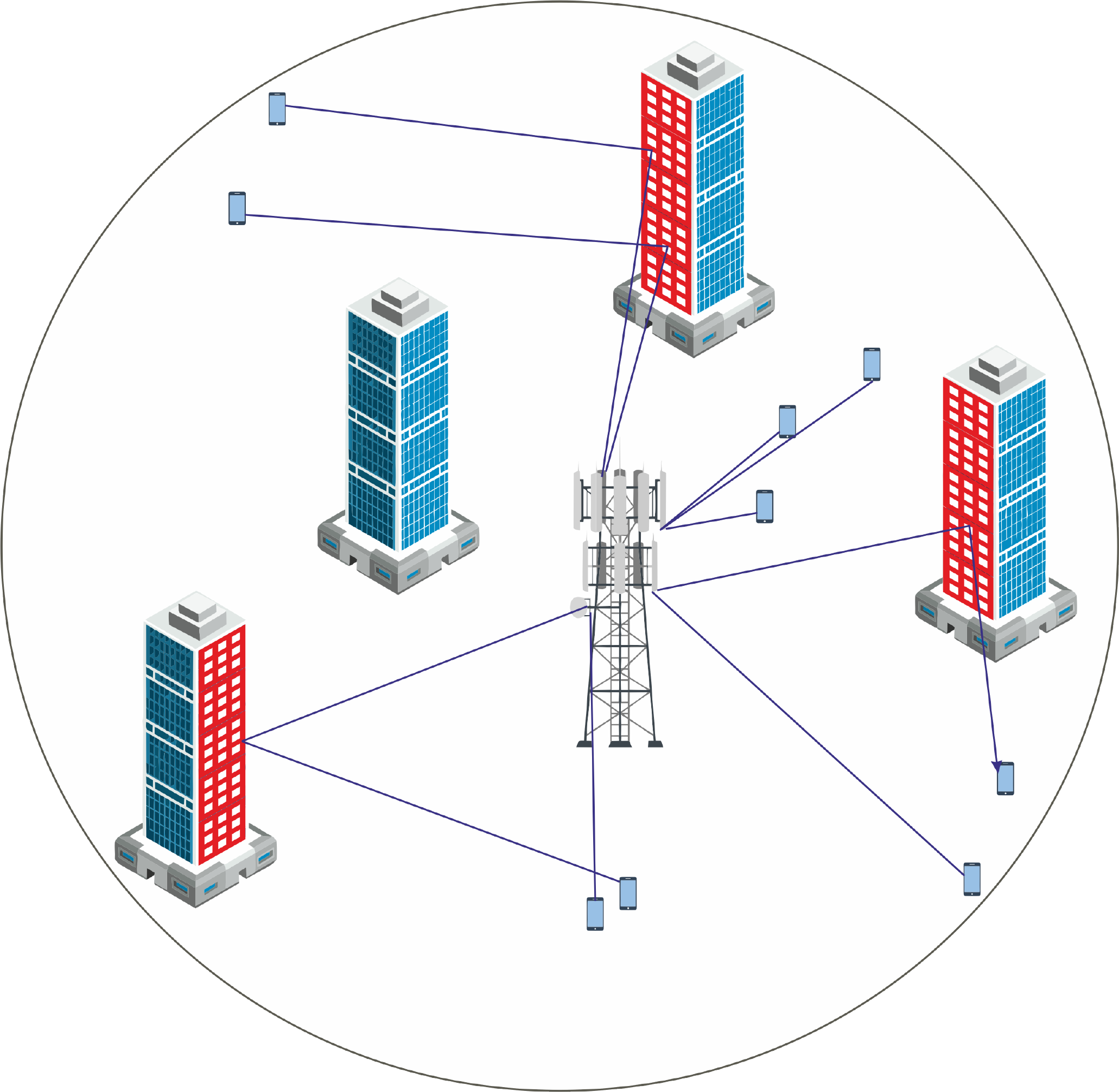}
     \label{fig.SystemModel}
     }
     \subfigure[]{\includegraphics[width=0.23\textwidth]{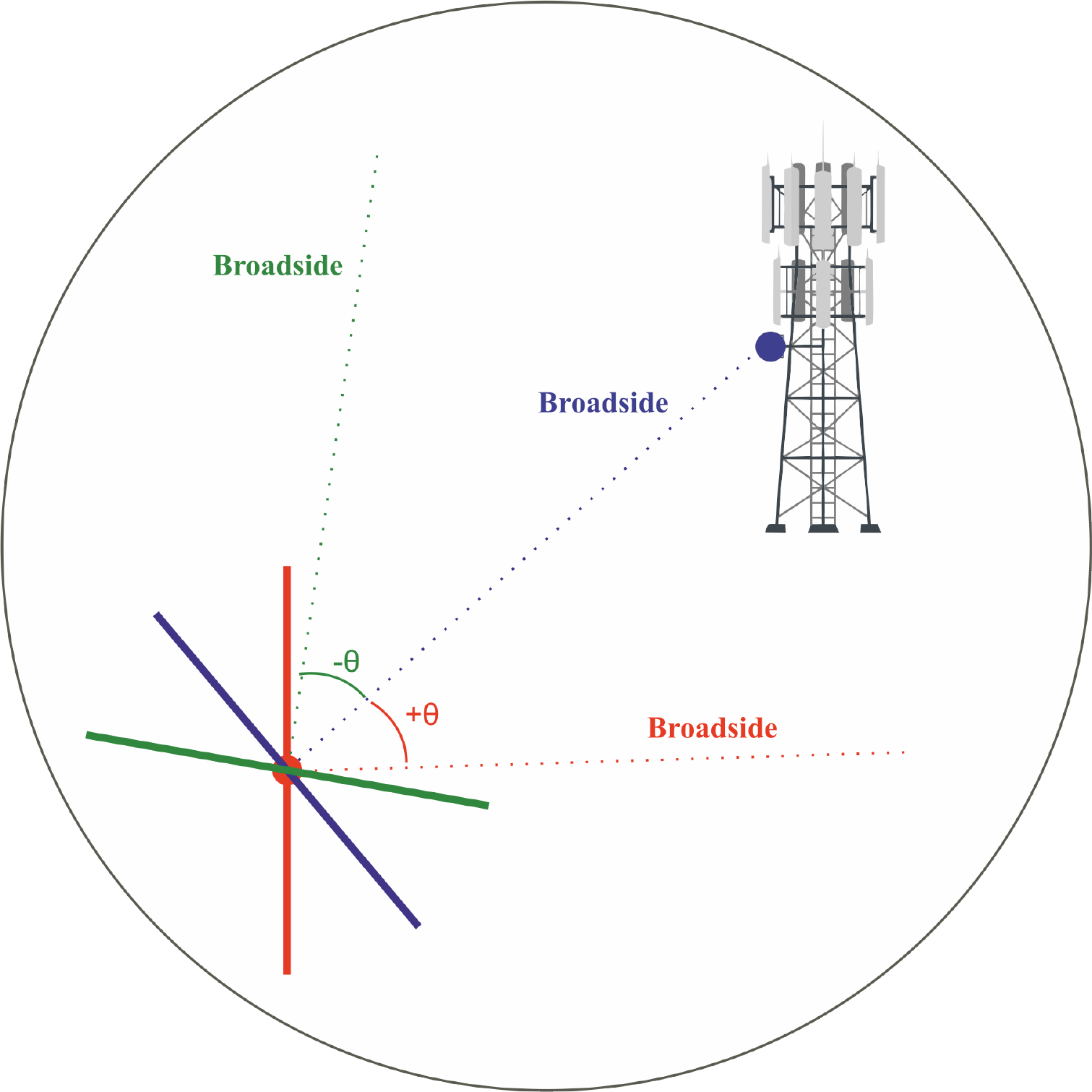}
     \label{fig.ValidRange}}
        \caption{(a) IRS-assisted system, (b) demonstration of valid azimuth orientations of IRS with respect to the BS.}
\end{figure}


We introduce the function $f(\cdot)$ which accepts a set $\mathcal{S}\subseteq \mathcal{P}$ of positions in 3D space, as locations of deployed IRSs, with $f(\mathcal{S})$ determining a subset of points in the region of interest $f(\mathcal{S})\subseteq \mathcal{L}$ in which UEs have LOS link to at least one of the points in $\mathcal{S}$ or to the BS. In this work, on the one hand, the objective is to reduce the placement cost, i.e., reduction in the number of deployed IRSs, while on the other hand, maximizing the extended LOS coverage provided by the BS and the deployed IRSs. Considering a maximum of $K$ deployed IRSs, the optimization problem can be formulated as follows:
\begin{equation}
    \begin{aligned}
        \underset{\mathcal{S}}{\max}~~ &f(\mathcal{S})\\
        \text{s.t.}~&|\mathcal{S}|\leqslant K\\
    \end{aligned}
    \label{eq.optorig}
\end{equation}
The optimization problem in \eqref{eq.optorig} can be generalized by considering different importance/priorities for the points within the area of interest which are denoted by $q_n,~ n\in\{1,\ldots,N\},$ and different costs of deployment for different locations of the deployment area $c_m,~ m\in\{1,\ldots,M\},$ assuming a total budget of deployment $B$. In the case of uniform cost, we can assume $c_m=1$, in which case the budget $B$ reduces to the maximum number of deployed IRSs $K$. We define $g(\mathcal{S})$ as the new objective function which is a weighted sum of the area covered exclusively by IRSs, i.e., the extended coverage provided by deployed IRSs set $\mathcal{S}$ that would not be covered otherwise with the BS (in a mathematical form, $f(\mathcal{S})\setminus f(\emptyset)$).
Then,
\begin{equation}
    \begin{aligned}
        \underset{\mathcal{S}\in \mathcal{P}}{\max}~~ &g(\mathcal{S}) = \sum_{n,~ \bL_n\in f(\mathcal{S})\setminus f(\emptyset)}{q_n}\\
        \text{s.t.}~& \sum_{s\in\mathcal{S}}c_s\leqslant B
    \end{aligned}
    \label{eq.optprob}
\end{equation}

Since, \eqref{eq.optprob} is a combinatorial optimization problem and NP-hard in nature \cite{9712623}, we prove submodularity of the problem and employ greedy optimization
as a suboptimal approach to tackle this problem.

\section{IRS Submodular Placement} \label{sec.IRSPlacement}





In this section, we prove the submodularity of the coverage function (i.e, $g(\cdot)$), which enables us to use a greedy algorithm to perform IRS placement with near-optimality guarantees.
The sketch of the proof is as follows. We start with proving the submodularity of the problem considering all the candidate locations combined with all feasible orientations. Then, we propose the IRS placement algorithm, in which, given the candidate set of locations, at each step, it selects the best location (in a greedy sense) with the best feasible orientation to add to the so far set of selected locations. This feature of the algorithm clarifies that even though we have only considered the set of locations (i.e., not combined with all feasible orientations), it is technically equivalent to the conditions of the provided proof, which consequently preserves the optimality of the orientation while reducing the computational complexity. In other words, the proposed algorithm provides a solution (i.e., a set of IRS locations and corresponding orientations) with near-optimal performance. Finally, some variations of the algorithm under special cases are presented.

\begin{theorem}\label{theorem.sub}
Assuming the set $\mathcal{P}$ contains all candidate locations of IRS combined with all feasible orientations, $g: 2^{\mathcal{P}} \rightarrow \mathbb{R}_+$ defined in \eqref{eq.optprob} is a normalized, monotone, submodular set function.
\end{theorem}
\begin{proof}
The proof is derived in Appendix \ref{ap.theorem1}.
\end{proof}

Based on Theorem \ref{theorem.sub}, $g(\cdot)$ with a ground set containing all candidate locations of IRS combined with all feasible orientations is a normalized, monotone, submodular function. 
In the sequel, to tackle \eqref{eq.optprob},  we first propose a near-optimal greedy algorithm for the uniform cost
problem and then proceed to the general case of IRS placement with non-uniform costs.

\subsection{Uniform Placement Cost for IRSs}
In this part, we consider uniform placement costs for IRSs, i.e., $c_m=1$, and assume the maximum number of deployed IRSs to be $K$. Leveraging the submodularity of the objective function, we propose a greedy placement algorithm for IRSs. The algorithm starts with an empty set of deployed IRSs $\mathcal{S}=\emptyset$, i.e., coverage only provided by the BS; and the set of feasible IRS locations $\mathcal{P}$. At each iteration, among the feasible IRS locations, we select the location with the maximum marginal coverage (i.e., combined with its best orientation) and continue the iterations until no feasible IRS location exists. In order to adapt to the condition of Theorem \ref{theorem.sub}, with a slight abuse of notation, we theoretically allow multiple repetitions of the same location with different orientations in the selected set $\mathcal{S}$ as it is allowed in Theorem \ref{theorem.sub}. However, in reality (also observed in simulations), considering the limited valid FoV, this usually does not happen and selected locations are unique. The proposed
algorithm is summarized in Algorithm \ref{GreedyAlgo1}. By Theorem \ref{theorem.sub}, $g(\cdot)$ is a normalized, monotone, submodular function. Therefore, Algorithm \ref{GreedyAlgo1} maximizes the objective function with a $1-1/\rm{e}$ optimality guarantee \cite{9186137,8537943}.

\renewcommand{\algorithmicrequire}{\textbf{Input:}}
\renewcommand{\algorithmicensure}{\textbf{Output:}}
\begin{algorithm}[t]\caption{\sc Uniform Cost IRS Placement Algorithm}
	\label{GreedyAlgo1}
	\begin{algorithmic} [1]
		\REQUIRE $g(\cdot)$, $K$, $\mathcal{P}$, $\mathcal{N}$.
		\ENSURE $\mathcal{S}$.
		
		{\bf Initialization:} $\mathcal{S} = \emptyset$;	
				
		\FOR {$k = 1$ to $K$} 		
		\STATE $\text{Determine the optimal orientation for each } p\in\mathcal{P}$;
		\STATE $p^* = \underset{p\in \mathcal{P}}{\arg\max}~ g(\mathcal{S}\cup\{p\})$, $\mathcal{S} \gets \mathcal{S}\cup \{p^*\}$;
		\ENDFOR		
	\end{algorithmic}
\end{algorithm}

\subsection{Nonuniform Placement Cost for IRSs}


In this part, we consider nonuniform placement cost for IRSs, so that $c_m$s are some given real numbers. We assume $B$ to be the maximum deployment budget and hence we have $\sum_{s\in\mathcal{S}}c_s\leqslant B$. 
The proposed algorithm is similar to Algorithm \ref{GreedyAlgo1} with the slight difference that a set of feasible unselected IRSs locations $\mathcal{B}$, i.e., the locations that do not exceed the budget constraint $B$, is considered, and the algorithm continues the iterations until no feasible IRS location exists. The proposed algorithm is summarized in Algorithm \ref{GreedyAlgo2}.


\begin{algorithm}[t]\caption{\sc Nonuniform Cost IRS Placement Algorithm}
	\label{GreedyAlgo2}
	\begin{algorithmic} [1]
		\REQUIRE $g(\cdot)$, $K$, $\mathcal{P}$, $\mathcal{N}$.
		\ENSURE $\mathcal{S}$.
		
		{\bf Initialization:} $\mathcal{S} = \emptyset$, $\mathcal{B} = \{p|p\in\mathcal{P}, c_p\leqslant B \}$
				
		\WHILE{$|\mathcal{B}|>0$}
		\STATE \text{Determine the optimal orientation for each } $p\in\mathcal{P}$
		\STATE $p^* = \underset{p\in \mathcal{P}}{\arg\max}~ g(\mathcal{S}\cup\{p\})$ 
		\STATE $\mathcal{S} \gets \mathcal{S}\cup \{p^*\}$, $B = B - c_{p^*}$, $\mathcal{B} = \{p|p\in\mathcal{P}, c_p\leqslant B \}$;
		\ENDWHILE		
	\end{algorithmic}
\end{algorithm}

It is shown that Algorithm \ref{GreedyAlgo2} can perform arbitrarily bad, as it ignores differences in placement cost \cite{9186137}. In fact, this algorithm might make a mistake by adding an expensive IRS location with a marginal coverage improvement with respect to the other IRS locations. In such examples, it is possible to select more IRS locations with the same total cost and a greater total coverage. To overcome this weakness, we propose a coverage-cost greedy algorithm that jointly considers placement cost and coverage of IRS locations. The gain-cost greedy algorithm is similar to Algorithm \ref{GreedyAlgo2} except that for adding a new IRS location, we select the IRS location with the maximum coverage-cost ratio as follows
\begin{equation}
    p^* = \underset{p\in\mathcal{B}}{\arg\max}\frac{g(\mathcal{S}\cup \{p\})-g(\mathcal{S})}{c_p}.
\end{equation}

Unfortunately, even the coverage-cost algorithm can perform arbitrarily bad and there is no performance guarantee for this algorithm. However, running both the nonuniform cost and coverage-cost algorithms and selecting the best of them, provides a $\frac{1-1/\rm{e}}{2}$ optimality guarantee \cite{9186137}. It should be noted that there is a more computationally complex algorithm that achieves a better optimality guarantee which enumerates all feasible IRS placement subsets with the cardinality of $3$ and separately augments each of them with a coverage-cost algorithm \cite{9186137,9387156}. For the sake of conciseness, we omit the details of this algorithm here.

\subsection{Optimal Orientation Selection}
In the proposed placement algorithms, investigating the candidate locations for the IRSs, we need to select the best orientation for the location under investigation. One way to do so is gridding the azimuth angle domain and run an exhaustive search to determine the optimal orientation, i.e., the orientation that provides the maximum coverage. However, here we propose a low complexity, yet optimal algorithm to determine the surface orientation, which works in the continuous domain (i.e., gridless). The algorithm is as follows. For the given IRS location, we determine the valid orientation range with respect to the BS denoted as $[\theta^B_l,\theta^B_r]$, and for every position in the area of interest that is in LOS of this location under investigation and is a candidate to be exclusively covered by this IRS, i.e., not covered either by the BS or the selected IRSs so far, we denote the valid orientation by $[\theta_l^d,\theta_r^d],~ d\in\{1,\ldots,D\}$, with $D$ being the total number of candidate UE locations for this specific IRS location. It should be noted that any orientation angle $\theta\in[\theta_l,\theta_r]$ satisfies that specific constraint. Then, a list $\ba\in\mathbb{R}^{(2D+2) \times 1}$ is formed stacking all these angles in an ascending order with $2D+2$ entries. We define a counter list of the same length, i.e., counting the number of overlapping orientation ranges at each angle, represented as $\bt\in\mathbb{R}^{(2D+2)\times 1}$ initialized by zeros and setting $t_1 = 1$. Going through the entries of the list one by one, at entry $i\in\{2,\cdots,2D+2\}$, $t_i$ is $t_{i-1}$ plus one if the $i$th entry of $a_i$ is the left extreme of a valid orientation range and minus one otherwise. Finally, among the entries of the list that are in the range of $[\theta^B_l,\theta^B_r]$ (i.e., ensuring to be a valid orientation toward the BS), orientations in the range starting with the angle with corresponding maximum counter in $\bt$, denoting the index with $i_{\text{max}}$, until the next entry of the list, i.e., $i_{\text{max}}+1$, provide the best orientations of the surface with $t_{i_{\text{max}}}-1$ presenting the number of additional covered points with this IRS (subtraction of one is due to the fact that BS has been counted in this counter).
The proposed algorithm is summarized in Algorithm \ref{GreedyAlgo3}.

\renewcommand{\algorithmicrequire}{\textbf{Input:}}
\renewcommand{\algorithmicensure}{\textbf{Output:}}
\begin{algorithm}[t]\caption{\sc Orientation Selection Algorithm}
	\label{GreedyAlgo3}
	\begin{algorithmic} [1]
		\REQUIRE $\mathbf{\theta}_l =[\theta_l^B, \theta_l^1,\ldots,\theta_l^D]$,
         $\mathbf{\theta}_r =[\theta_r^B, \theta_r^1,\ldots,\theta_r^D]$
		\ENSURE $\text{Coverage Count}, \theta_{\text{max}}$.
		
		{\bf Initialization:} $\ba = \text{sort}(\mathbf{\theta}_l\cup \mathbf{\theta}_r)$,
		\bt = [1,0,\ldots,0]		
		\FOR {$i = 2$ to $2D+2$} \IF{$a_i \in \mathbf{\theta}_l$}
        \STATE $t_i = t_i + 1$
        \ELSE
        \STATE $t_i = t_i - 1$
        \ENDIF
		\ENDFOR
        \STATE $i_{\text{max}} = \underset{i,~ \theta_l^B\leqslant a_i \leqslant \theta_r^B}{\arg\max} t_i$;
        \STATE $\text{Coverage Count} = t_{i_{\text{max}}}-1,~ \theta_{\text{max}} = a_{i_{\text{max}}}$;
	\end{algorithmic}
\end{algorithm}

\begin{theorem}\label{theorem.orientation}
One of the $\theta_l$s corresponding to the BS or UE positions is among the optimal orientations of the surface at the given location.
\end{theorem}
\begin{proof}
The proof is derived in Appendix \ref{ap.theorem2}.
\end{proof}

Since the proposed algorithm considers all the $\theta_l$s and leveraging Theorem \ref{theorem.orientation}, finding one of the optimal orientations is proved.

\subsection{Computational Complexity Discussion}
We calculate the computational complexity assuming $K$ selected IRSs (in case of nonuniform placement cost, the final selected number of IRSs is not clear in advance, however, one can get a rough estimate based on the budget and average cost of placement).
The algorithm has $K$ iterations. At each iteration, for each IRS candidate location (i.e., $O(M)$), we determine the valid range of orientations of the surface for the area of interest (i.e., $O(N)$), search the list of angles (i.e., $O(N\log N)$), select the optimal orientation (i.e., $O(N)$), and finally select the optimum location (i.e., $O(M)$). Therefore, the total computational complexity of the algorithm is $O(KM(N+N\log N+M))$, considering the fact that typically $M<< N$, the total computational complexity reduces to $O(KMN\log N)$.
\section{Simulation Results}

In this section, we investigate the application of the submodular algorithm proposed in this paper for the placement and orientation of IRSs in a realistic environment using simulation-generated LOS coverage maps.
For this purpose, we use the channel model QuaDRiGa \cite{Quad} developed at Fraunhofer HHI, which has recently been extended by the feature of LOS-detection.
This allows site-specific simulations by incorporating data from the environment.
For the simulations in this paper, the data includes 3D-information of buildings and height profiles of the terrain under consideration, which is the campus network ``Berlin 5G testbed'', shown in Fig. \ref{fig.layout}.
The location of the BS, $\bP_B$, is on top of the HHI building.

\begin{figure}[ht]
     \centering
     \includegraphics[width=0.48\textwidth]{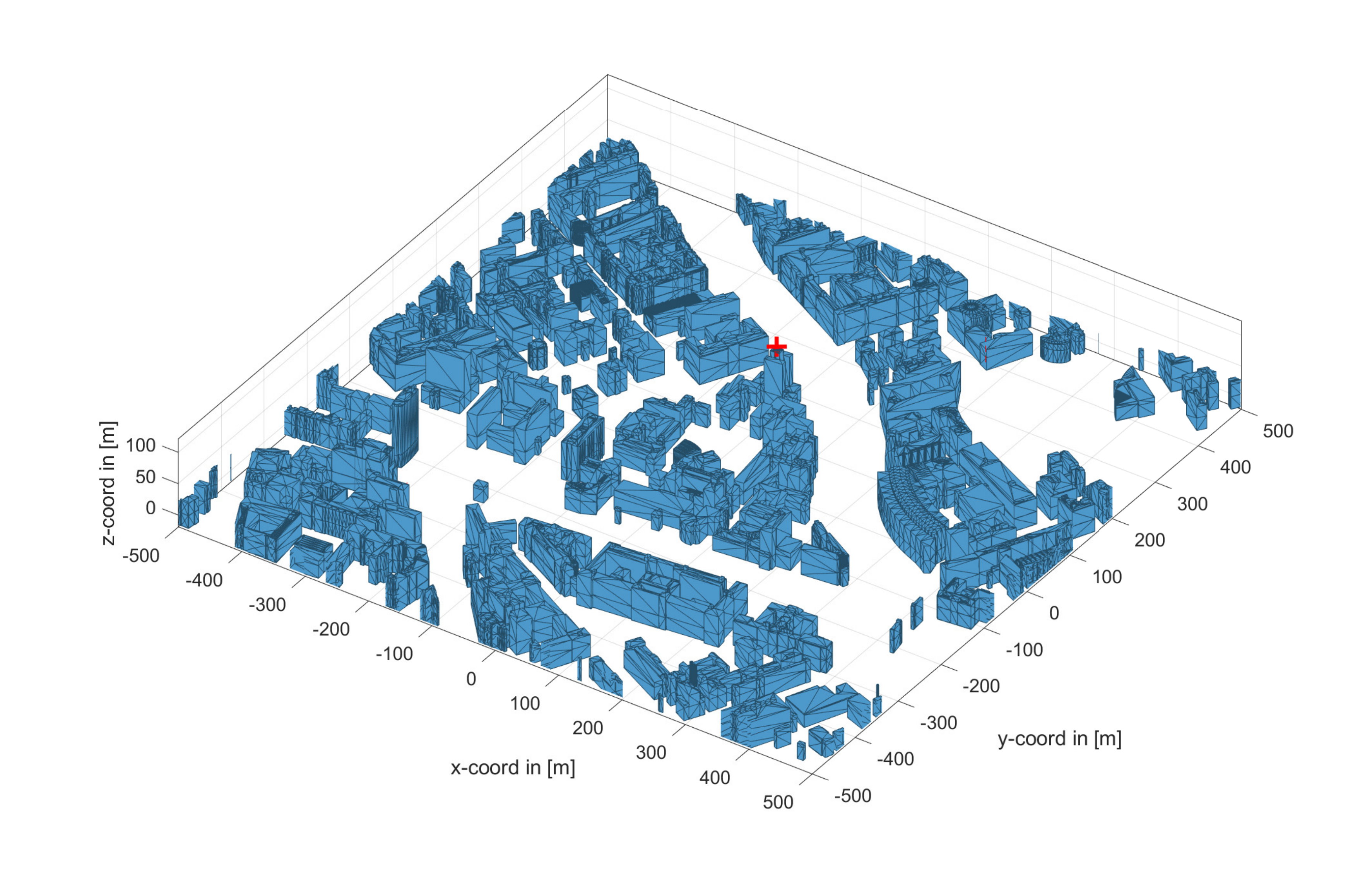}
     \caption{Three-dimensional simulation layout of the ``5G Berlin testbed'' with a BS placed on top of the HHI building.}
     \label{fig.layout}
\end{figure}

The channel model operates on a position-based sample grid, where it is determined for each point $\bL\in\mathcal{L}$ whether an LOS path exists.
From this, LOS coverage maps can be generated.
The simulation assumptions are summarized in Table \ref{tab.simulation}.

{\footnotesize
\begin{table}[t]
  \caption{Simulation settings}
  \label{tab.simulation}
  \centering
  \begin{tabular}{cc}
    \toprule
    Simulation tool  & QuaDRiGa \cite{Quad} \\
    LOS detection & 3D geometric\\
    Simulation area & 5G Berlin testbed \\
    Map size & $\SI{1000}{\metre} \times \SI{1000}{\metre}$\\
    \midrule
    BS height & $\SI{63}{\metre}$ \\
    UE height & $\SI{1.5}{\metre}$ \\
    IRS height & $\SI{30}{\metre}$ \\
    \midrule
    Grid size (display) & $\SI{2}{\metre} \times \SI{2}{\metre}$\\
    Grid size (optimization)& $\SI{10}{\metre} \times \SI{10}{\metre}$\\
    FoV $[-\theta_{\text{max}},+\theta_{\text{max}}]$ & $[\ang{-60}, \ang{60}]$\\
    \bottomrule
  \end{tabular}
\end{table}
}
  
In Fig. \ref{fig.BSUECoverage}, the LOS coverage $f(\emptyset)$ for the direct link from the BS to the UEs is shown, i.e., only UEs at orange positions have an LOS path to the BS.
The aim is therefore to expand the coverage area by smart positioning and orientation of IRSs and to obtain increased coverage through virtual LOS paths.
This requires that the path from BS to IRS also has LOS, which is shown in Fig. \ref{fig.BSIRSCoverage} for an IRS height of \SI{30}{\metre}.
Accordingly, all the colored positions are potential IRS positions, $\bP\in\mathcal{P}$.

\begin{figure}[ht!]
     \centering
     \subfigure[]{\includegraphics[width=0.23\textwidth]{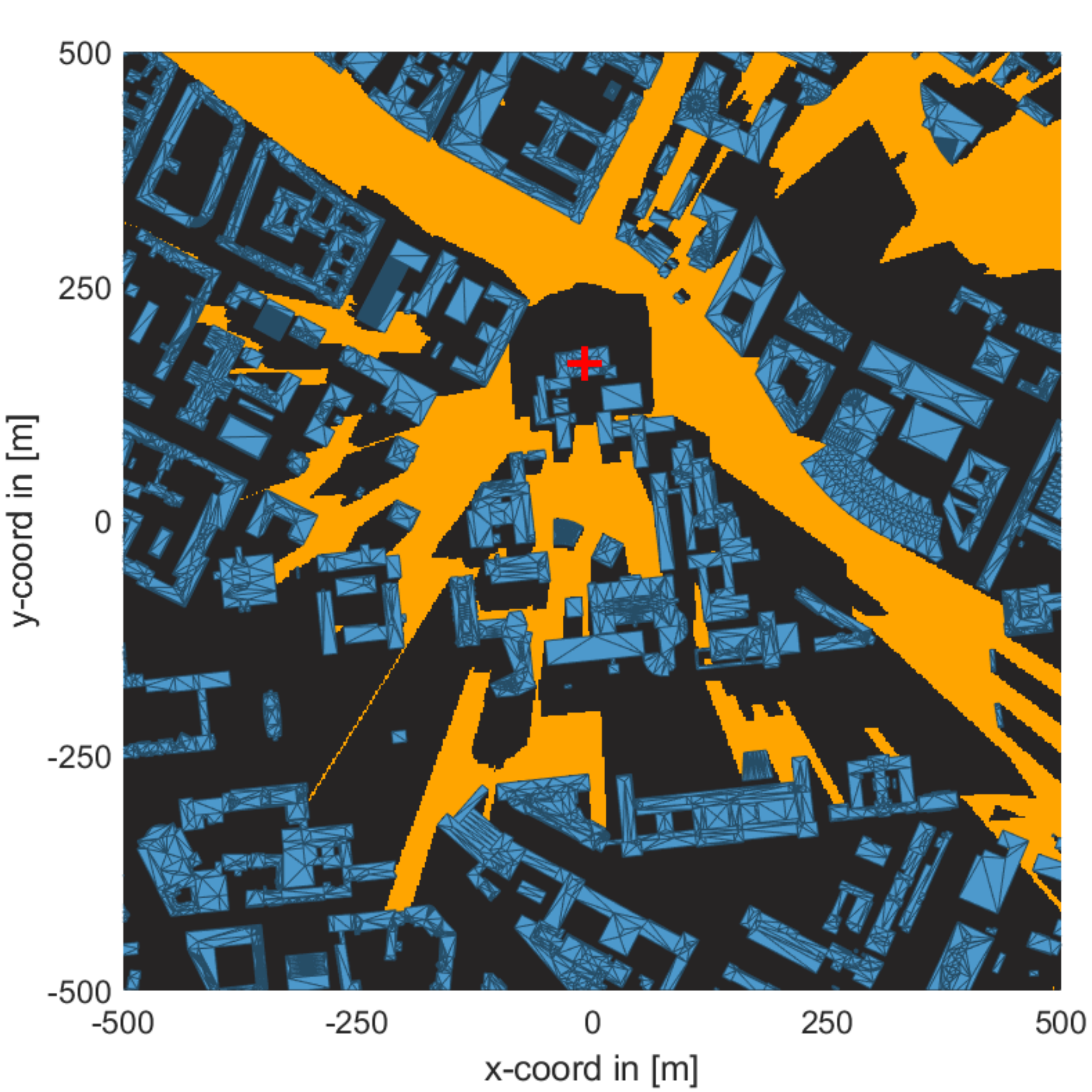}
     \label{fig.BSUECoverage}}
     \subfigure[]{\includegraphics[width=0.23\textwidth]{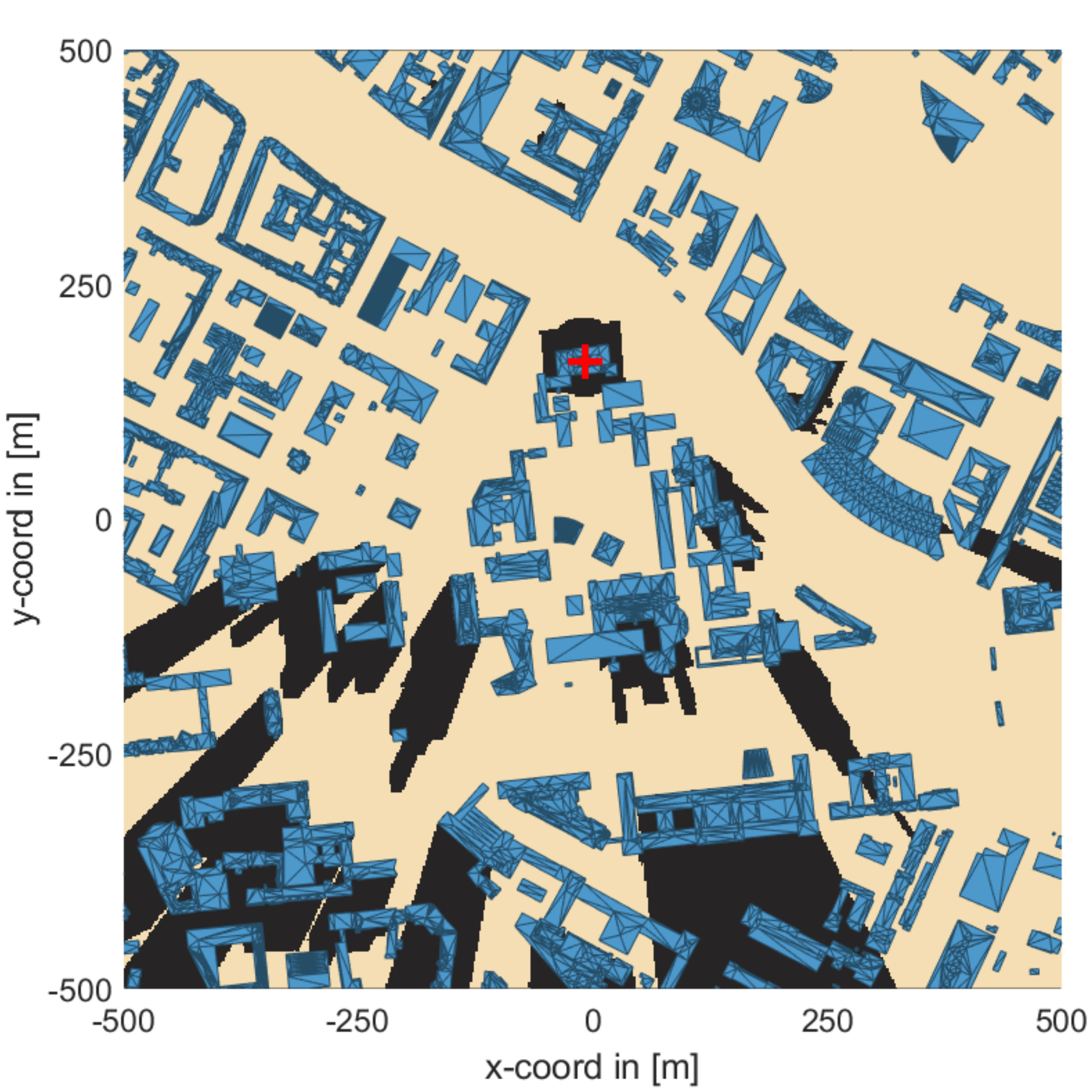}
     \label{fig.BSIRSCoverage}}
    \caption{LOS coverage provided by the BS where (a) shows direct LOS paths between BS and UE and (b) shows LOS paths between BS and potential IRS positions at \SI{30}{\metre} height.}
\end{figure}

First, we consider a small-scale scenario in which a set $\mathcal{P} = \{\bP_1,\ldots,\bP_5\}$ of $5$ IRS candidate positions (each at a height of \SI{30}{\metre}) is available.
Their locations and respective coverage are shown in Fig. \ref{fig.IRSCoverage}, where the grid size is $\SI{2}{\metre} \times \SI{2}{\metre}$.
White bars represent the IRS (not the right scale with respect to the buildings) and the respective total colored areas indicate the LOS coverage by the corresponding IRS without taking into account the IRS orientation.
Darker colored areas show the LOS coverage with the IRS orientation that maximizes LOS coverage for an FoV of $[\ang{-60}, \ang{60}]$, i.e., $f(\bP_1), \ldots, f(\bP_5)$, and dashed white lines correspond to the respective broadside of the deployed IRS.

\begin{figure*}[ht!]
     \centering
     \subfigure[]{\includegraphics[width=0.18\textwidth]{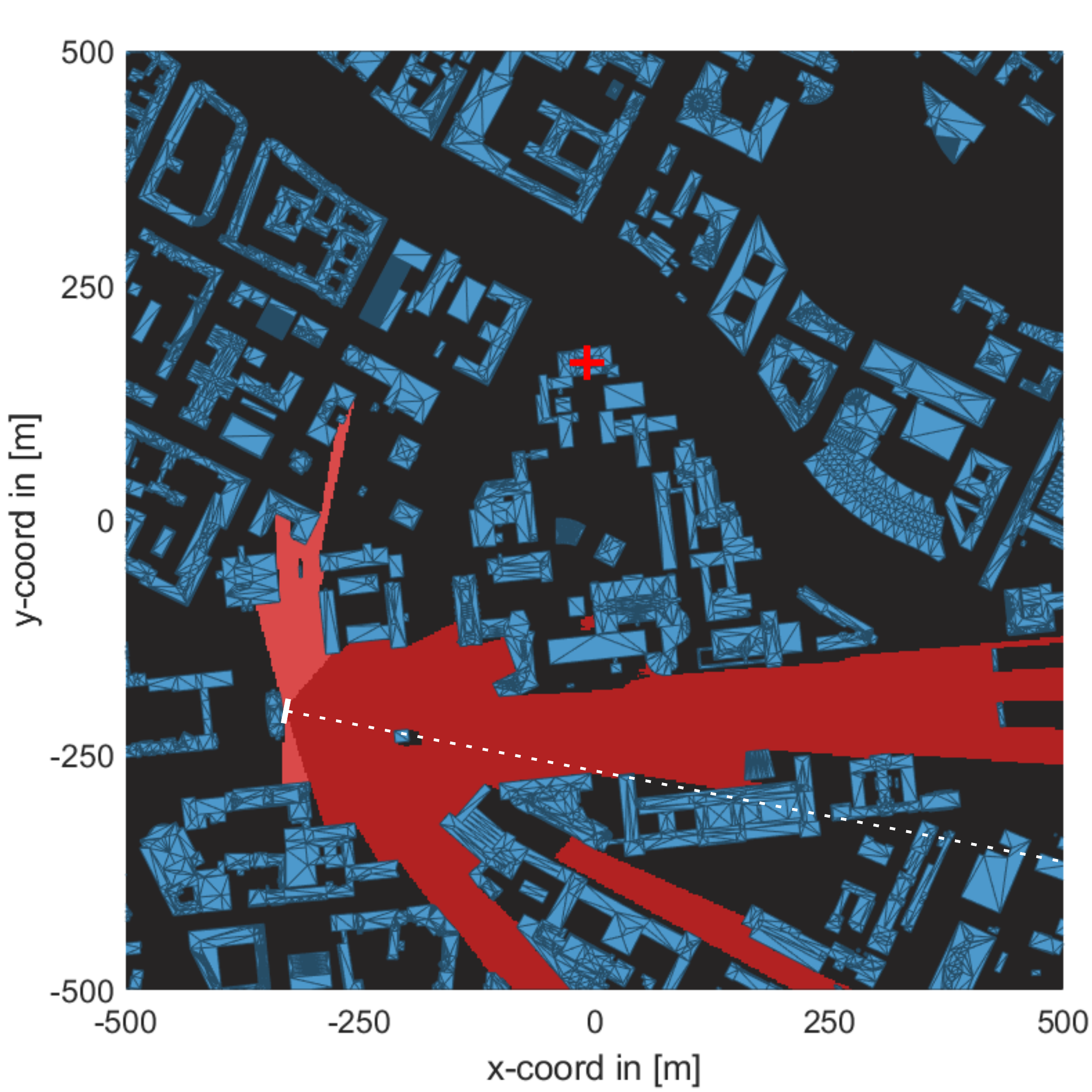}
     \label{fig.IRSCoverage1}}
     \subfigure[]{\includegraphics[width=0.18\textwidth]{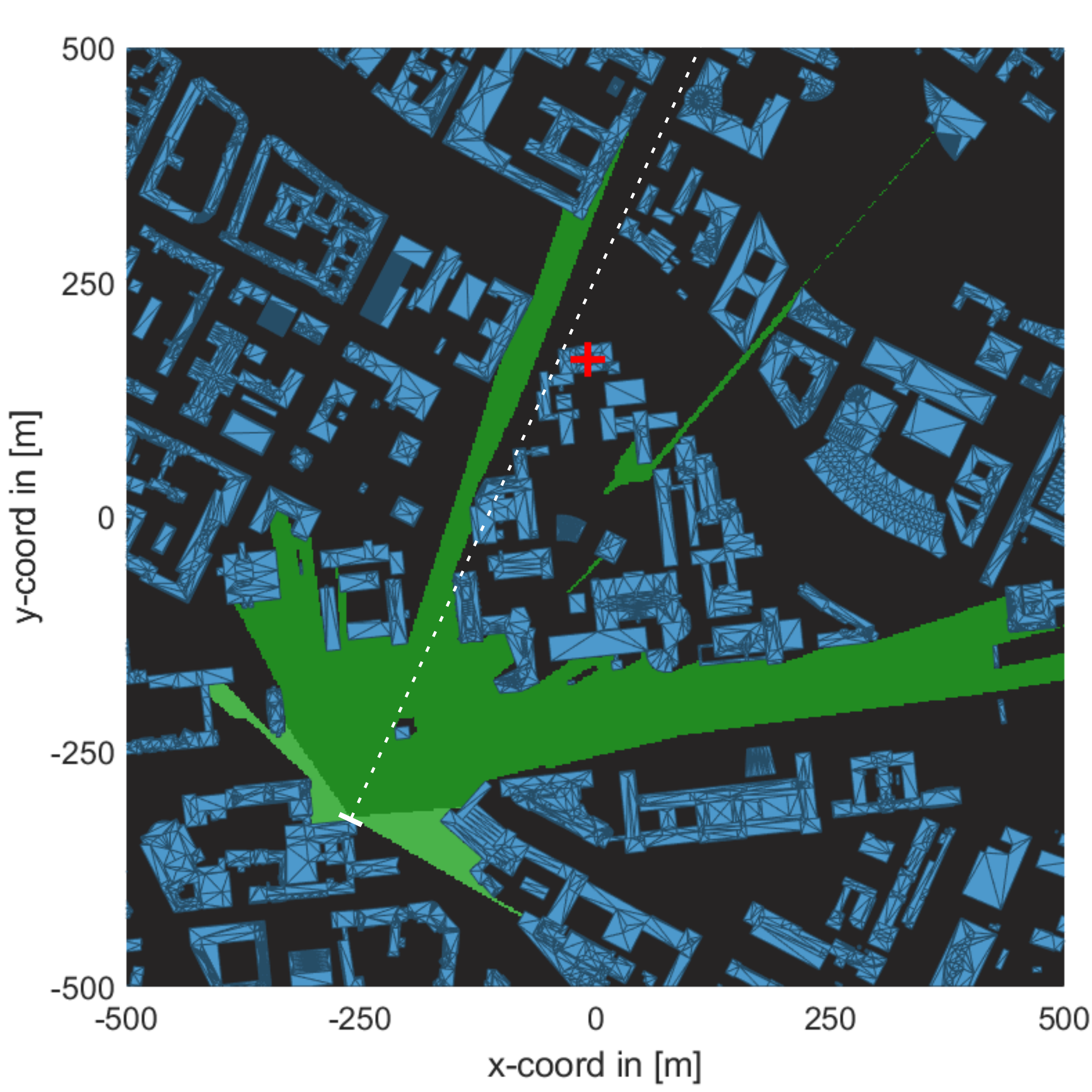}
     \label{fig.IRSCoverage2}}
     \subfigure[]{\includegraphics[width=0.18\textwidth]{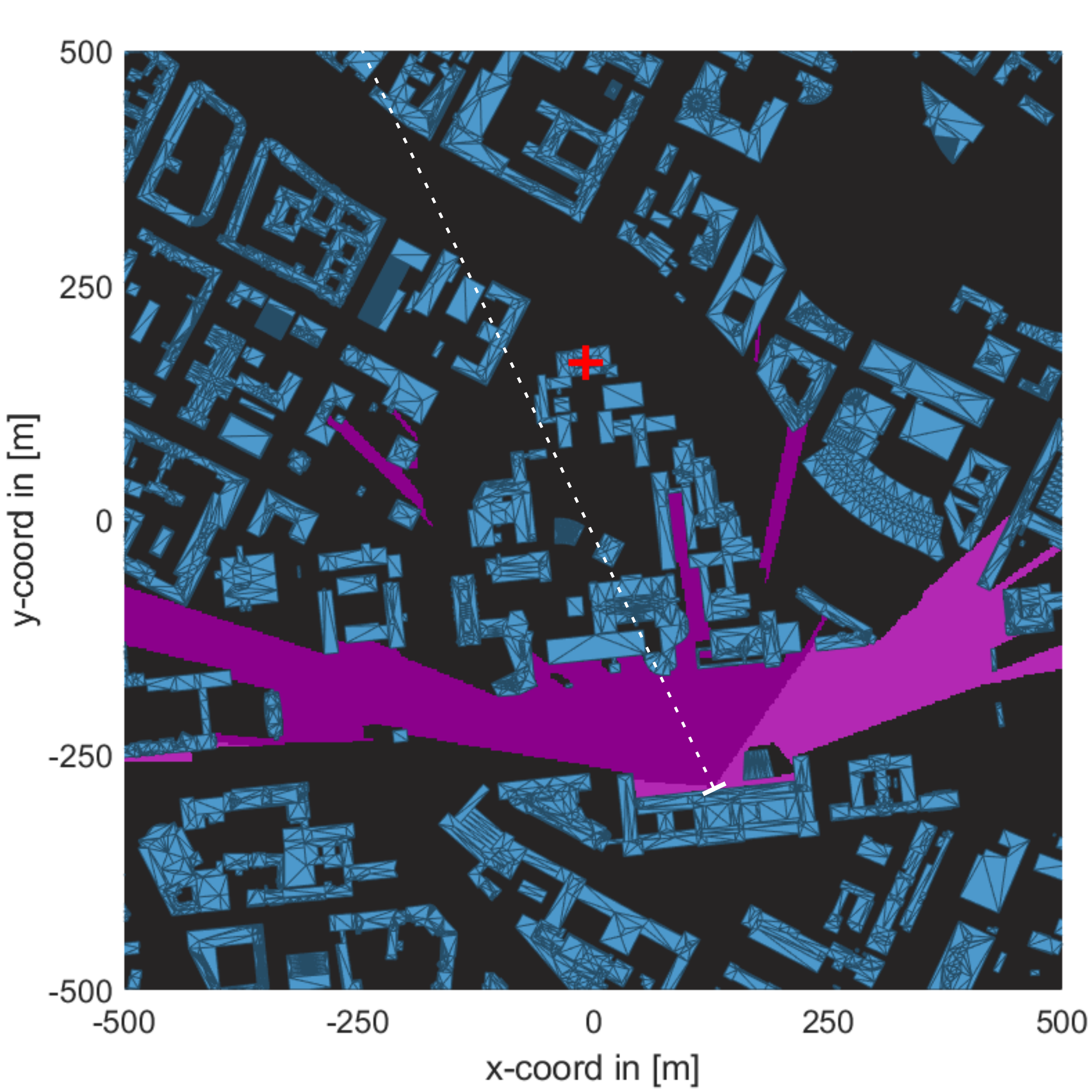}
     \label{fig.IRSCoverage3}}
     \subfigure[]{\includegraphics[width=0.18\textwidth]{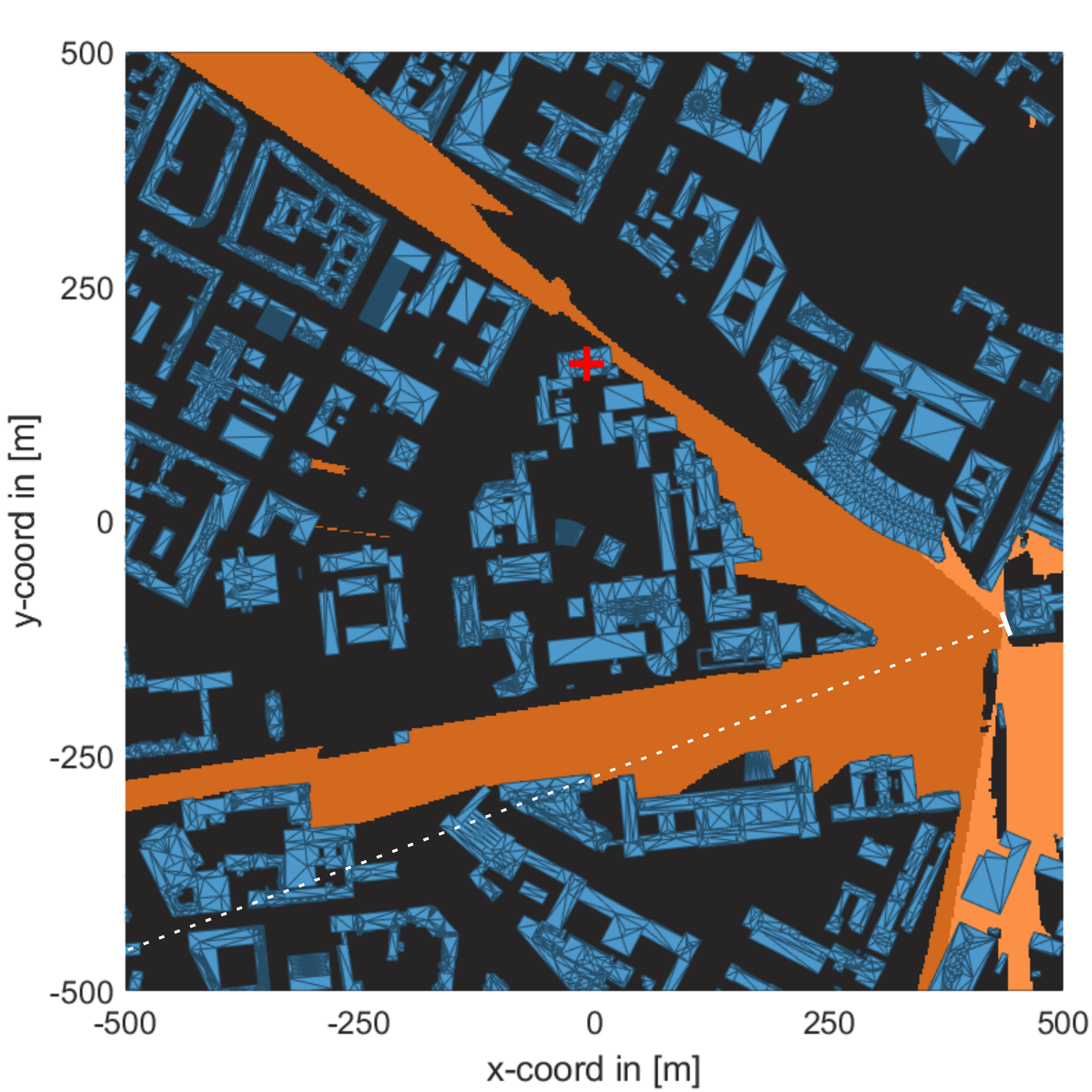}
     \label{fig.IRSCoverage4}}
     \subfigure[]{\includegraphics[width=0.18\textwidth]{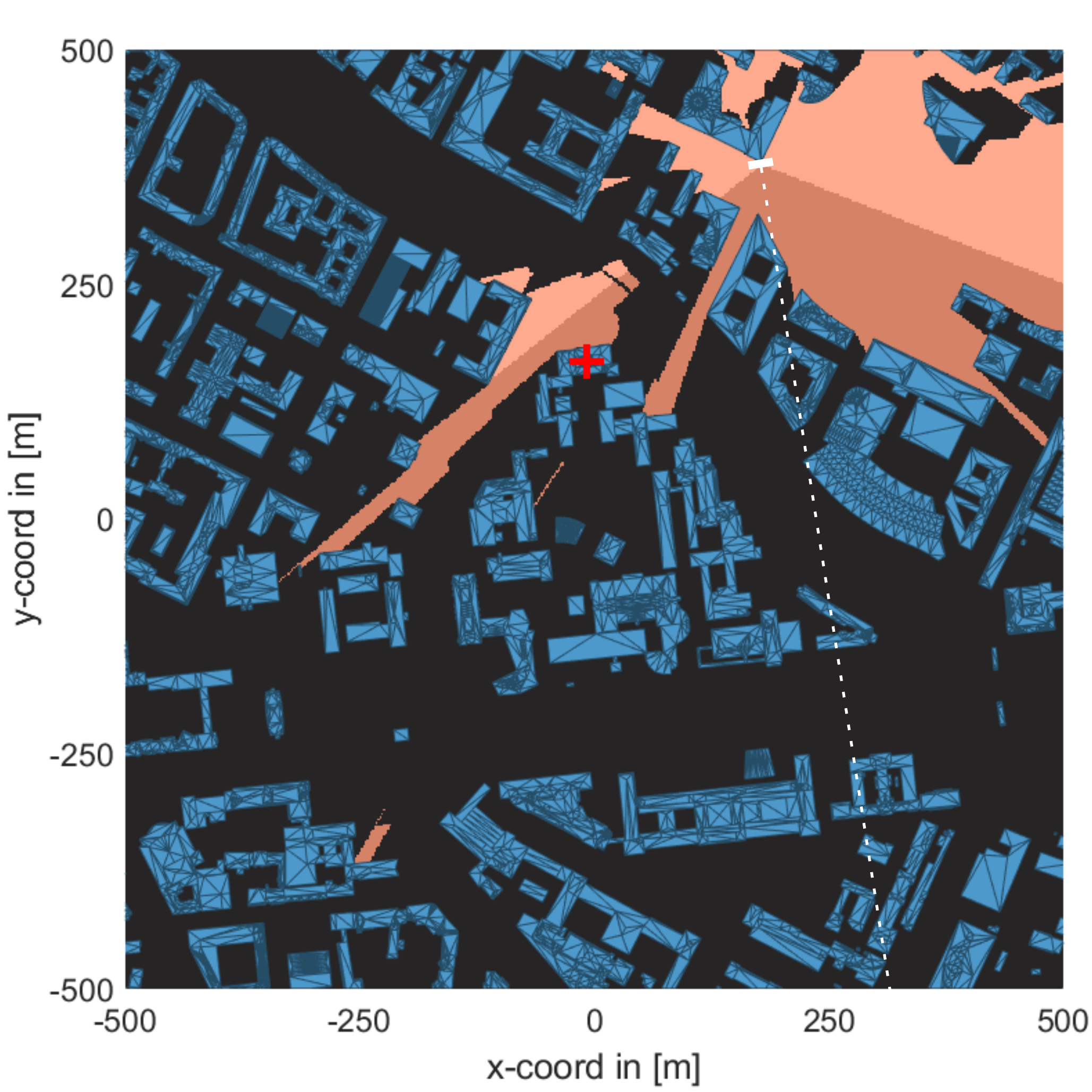}
     \label{fig.IRSCoverage5}}
    \caption{LOS coverage provided by $|\mathcal{P}| = 5$ candidate IRS positions.}
    \label{fig.IRSCoverage}
\end{figure*}

The proposed submodular IRS placement approach can now be used to determine multiple IRS positions and orientations that maximizes the LOS coverage provided by the BS from Fig.~\ref{fig.BSUECoverage}.
Note that this may result in orientations deviating from Fig. \ref{fig.IRSCoverage}.
In this example, for the sake of simplicity, no differentiation is made between the importance/priority of the points and the cost of deployment (uniform cost placement), i.e., $q_n = 1$ and $c_m = 1$.
The deployment for the optimal IRS, i.e., the one that extends the coverage the most, is shown in Fig. \ref{fig.IRSCoverage_final1}.
In Fig. \ref{fig.IRSCoverage_final2}, the second IRS is deployed which optimally (in a greedy sense) extends the superposition of the LOS coverage areas of the BS and the first IRS.

\begin{figure}[ht!]
     \centering
     \subfigure[]{\includegraphics[width=0.23\textwidth]{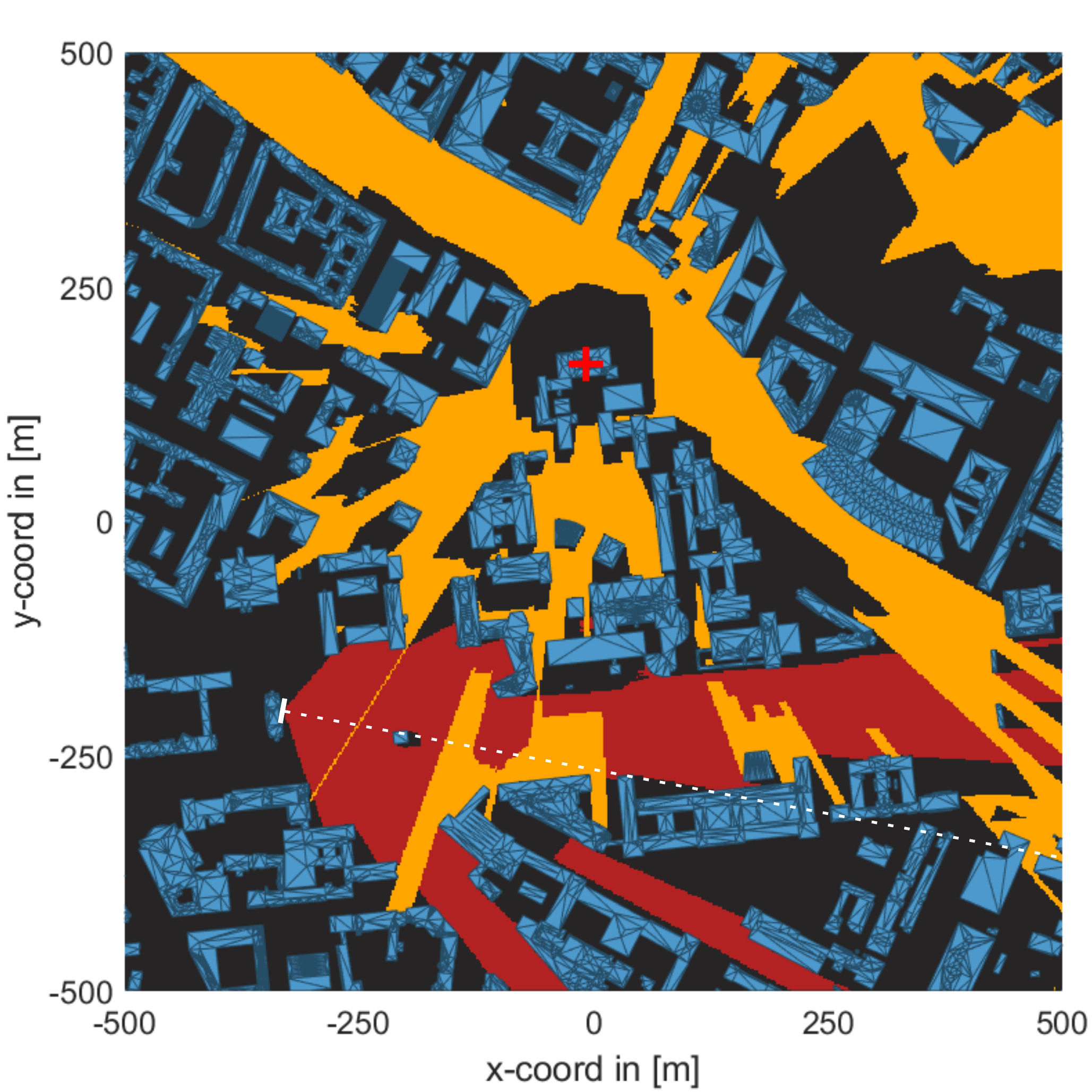}
     \label{fig.IRSCoverage_final1}}
     \subfigure[]{\includegraphics[width=0.23\textwidth]{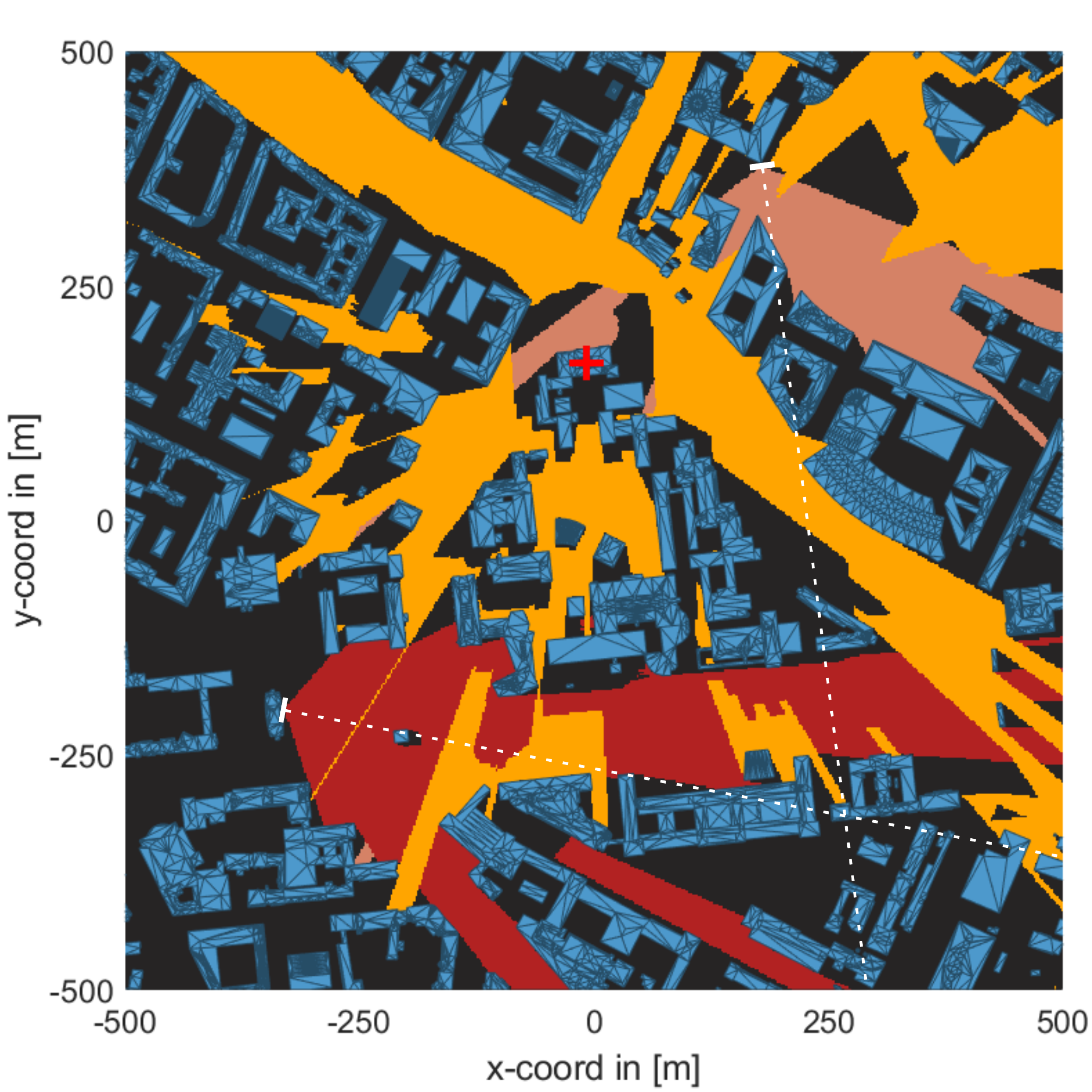}
     \label{fig.IRSCoverage_final2}}
        \caption{LOS coverage provided by BS and IRS for (a) $K=1$ and (b) $K=2$.}
        \label{fig.IRSCoverage_final}
\end{figure}

The presented submodular algorithm is near-optimal as shown in \ref{sec.IRSPlacement}.
However, an optimal determination of positions and orientations for multiple IRSs can be achieved with an exhaustive search.
The comparison of the two methods is provided in Fig. \ref{fig.IRSCoverage_exhaustive}, where LOS coverage, i.e., the number of obtained LOS points $f(\mathcal{S})$, is plotted for different numbers $K$ of deployed IRSs.
Thus, the coverage for $K = 0$ corresponds to the coverage provided by the BS only, i.e. $f(\emptyset)$.
Note that, only the $5$ potential IRS positions from the small-scale example, shown in Fig. \ref{fig.IRSCoverage}, are included and in order to make the computational effort of exhaustive search feasible, we calculate the orientations with an angular resolution of $5$ degrees.
For a fair comparison, this is also applied to the submodular approach, although no additional computational complexity regarding orientation is imposed by the submodular method as it is gridless.
We also set the maximum number of placed IRSs and the grid size to $K=4$ and to $\SI{10}{\metre} \times \SI{10}{\metre}$, respectively, for complexity reasons.
The results show that the submodular algorithm can extend coverage through IRSs almost as effectively as an exhaustive search for the considered scenario (For $K=0$ and $K=1$ it is obviously the same, for higher $K$, the difference is only a few points). 

\begin{figure}[ht!]
     \centering
     \includegraphics[width=0.45\textwidth]{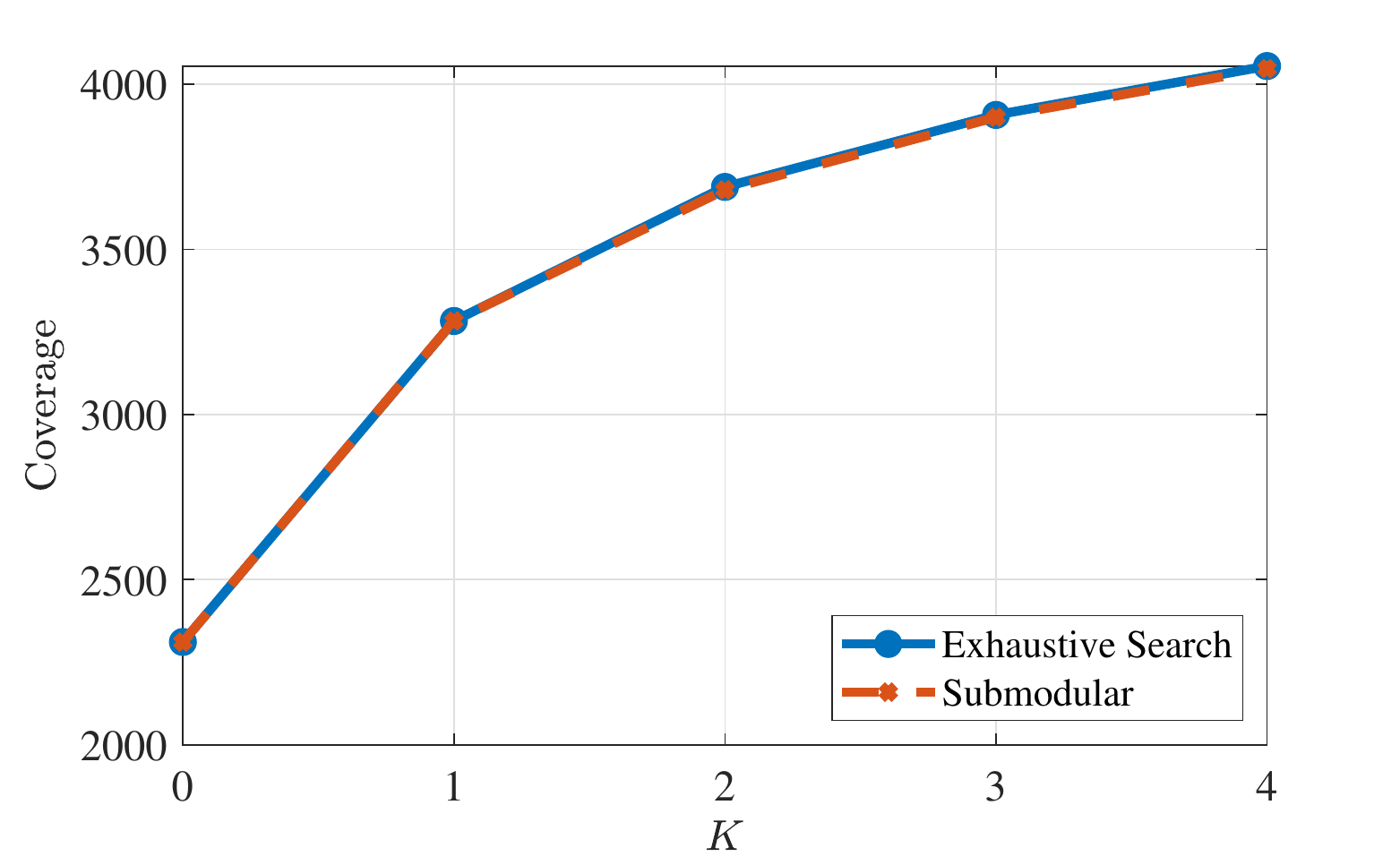}
        \caption{LOS coverage over number of placed IRSs $K$ with exhaustive search and proposed submodular algorithm.}
        \label{fig.IRSCoverage_exhaustive}
\end{figure}


In contrast to the exhaustive search that becomes computationally infeasible by increasing the size of the scenario, the proposed submodular algorithm can deal with large-scale IRS placement problems.
Fig. \ref{fig.IRSCoverage_submodular} shows the coverage $f(\mathcal{S})$ versus the number of IRSs $K$ for the submodular approach.
The potential IRS positions considered here correspond to all LOS coverage points at \SI{30}{\metre}, shown in Fig. \ref{fig.BSIRSCoverage}, and the grid size is $\SI{10}{\metre} \times \SI{10}{\metre}$ as before.
Further restrictions, for example, that no IRS may be placed on streets, can be included but are neglected in this example for simplicity.
In any case, it could be observed that the algorithm places IRSs on building facades, since this is where the view usually captures the most points, thus maximally expanding the coverage.
The results show a monotonically increase of coverage versus the number of deployed IRSs.
Also, the diminishing coverage is depicted as the coverage plot gradually saturates for high values of $K$.
In addition, for the present scenario, it is shown that coverage can be doubled with $K=5$ IRSs, i.e., $4637$ coverage points vs. $2311$ points with BS only.

\begin{figure}[ht]
     \centering
     \includegraphics[width=0.45\textwidth]{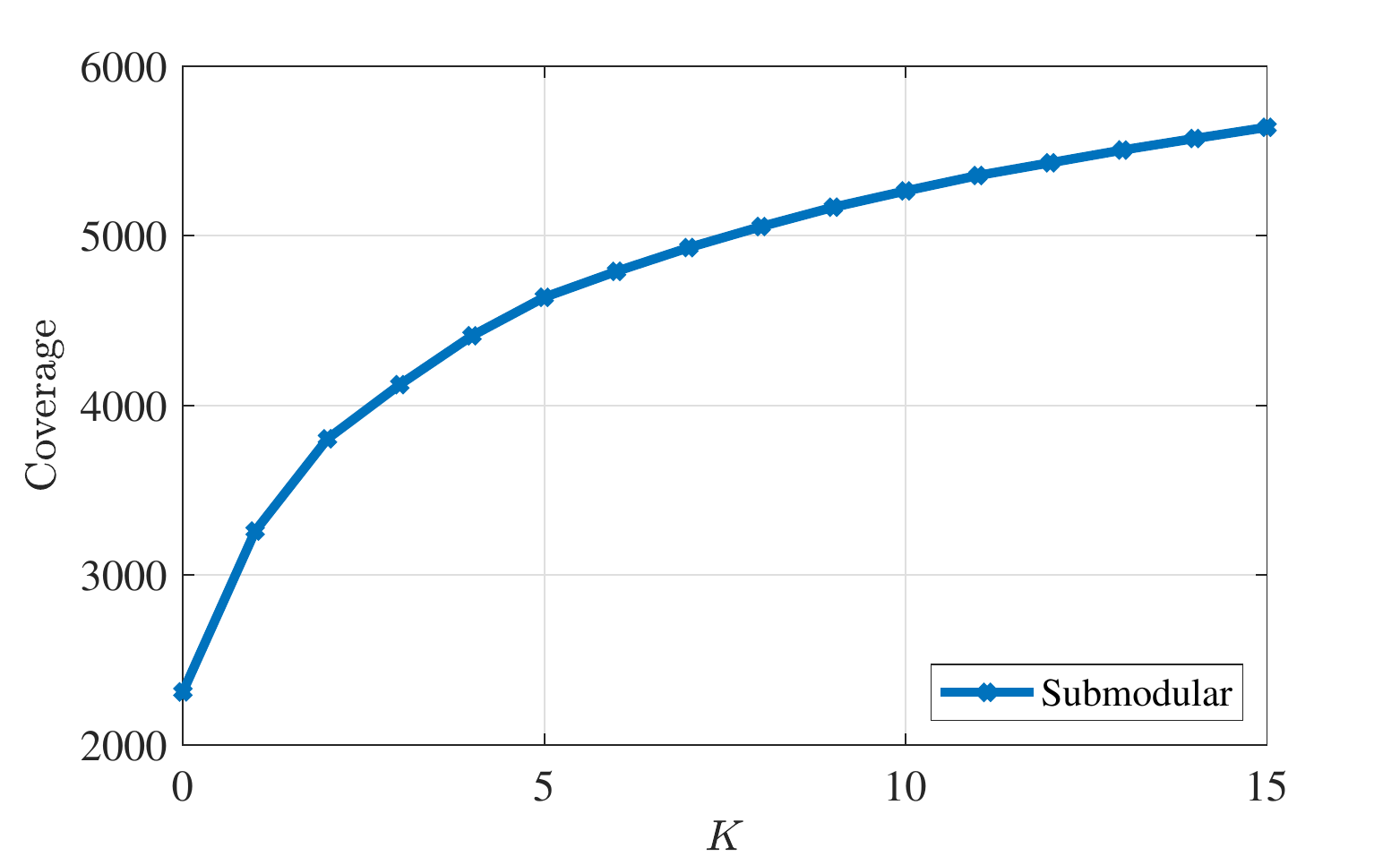}
        \caption{LOS coverage over number of placed IRSs $K$ with proposed submodular algorithm.}
        \label{fig.IRSCoverage_submodular}
\end{figure}

\section{Conclusions} 
In this paper, we study the IRS placement problem. We address the abstraction level for IRS modeling in the network planning setting and identify LOS and orientation as two key parameters. Considering the NP-hardness of the IRS placement problem, we propose an efficient algorithm with low computational complexity, yet with performance bound proved by the submodularity of the optimization problem. Enabling multiple IRSs placement in addition to the possibility to include a realistic outdoor environment are other contributions of the proposed algorithm. For future work, we will consider a scenario with multiple BSs and perform joint BS and IRS placement. In addition, we will generalize the LOS coverage to required RF powers and corresponding SNRs as the objective function taking into account transmit power, IRS size, and distance.

\section*{Appendices}
\subsection{Proof of Theorem \ref{theorem.sub}} \label{ap.theorem1}
First, we show that the function is normalized. That is, $g(\emptyset) = 0$. This can be proved noting that $f(\mathcal{S}) - f(\emptyset) = \emptyset$ for $\mathcal{S}=\emptyset$.
Now, we show the monotonicity of $g(\mathcal{S})$ in the following steps:
\begin{equation}
    \mathcal{A} \subseteq \mathcal{B} \rightarrow f(\mathcal{A}) \subseteq f(\mathcal{B})
    \label{eq.proof1}
\end{equation}
which is trivial as adding IRSs do not reduce the coverage. Consequently, assuming non-negative values for $q_n$s, i.e., $q_n \geqslant 0, n=1,\cdots,N$, \eqref{eq.proof1} leads to $g(\mathcal{A}) \subseteq g(\mathcal{B})$.

Finally, we show the submodularity of of the set function. To show submodualarity we need to prove that
\begin{equation}
    g(\mathcal{S}\cup\{a\}) - g(\mathcal{S}) \geqslant g(\mathcal{S}\cup\{a,b\}) - g(\mathcal{S}\cup\{b\})
\end{equation}
for any subset $\mathcal{S}\subseteq \mathcal{P}$ and arbitrary IRS locations $a,b\in\mathcal{P}$. Since $g(\mathcal{S})$ is a linear summation of $q_n$s for $\bL_n\in f(\mathcal{S})\setminus f(\emptyset)$ and $f(\mathcal{A})\subseteq f(\mathcal{B})$ for $\mathcal{A}\subseteq \mathcal{B}$, it is sufficient to show that
\begin{equation}
    f(\mathcal{S}\cup\{a,b\}) \setminus f(\mathcal{S}\cup\{b\}) \subseteq f(\mathcal{S}\cup\{a\}) \setminus f(\mathcal{S})
    \label{eq.theorem1step5}
\end{equation}
which we prove by contradiction. Let us assume that \eqref{eq.theorem1step5} is not satisfied which subsequently means that there exists an $x\in f(\mathcal{S}\cup\{a,b\}) \setminus f(\mathcal{S}\cup\{b\})$ in which $x \notin f(\mathcal{S}\cup\{a\}) \setminus f(\mathcal{S})$. Clearly, $x\in f(\mathcal{S}\cup\{a,b\}) \setminus f(\mathcal{S}\cup\{b\})$ means that $x\in f(\mathcal{S}\cup\{a,b\})$ and $x \notin f(\mathcal{S}\cup\{b\})$. Subsequently, $x\in f(\{a\})$ and $x\notin f(\mathcal{S})$ which results in $x \in f(\mathcal{S}\cup\{a\}) \setminus f(\mathcal{S})$, contradiction.

\subsection{Proof of Theorem \ref{theorem.orientation}} \label{ap.theorem2}
We prove the theorem by contradiction. Let us assume there is not a $\theta_l$ that provides the maximal coverage and consider $\theta^*$ as one of the optimal orientations for a given IRS location. Clearly, for the UE positions that are covered with this IRS location and the orientation $\theta^*$, we have $\theta_l \leqslant \theta^*$, as well as $\theta_l^B\leqslant \theta^*$ based on the assumptions. In the set of union of $\theta_l$s of covered UEs and BS, we select the maximal number, which here is denoted by $\theta_l^{\text{max}}$. It is straightforward to see that $\theta_l^{\text{max}}$, still satisfies all the orientation constraints introduced by the UEs and BS. Therefore, $\theta_l^{\text{max}}$ is also an optimal orientation and a left extreme points of the given orientation ranges. Contradiction and the proof is complete.


%


\ifCLASSOPTIONcaptionsoff
\newpage
\fi
\bibliographystyle{ieeetr}
\bibliography{ref}

\begin{thebibliography}{10}

\bibitem{wu2019intelligent}
Q.~Wu and R.~Zhang, ``Intelligent reflecting surface enhanced wireless network
  via joint active and passive beamforming,'' {\em IEEE Transactions on
  Wireless Communications}, vol.~18, no.~11, pp.~5394--5409, 2019.

\bibitem{pan2021reconfigurable}
C.~Pan, H.~Ren, K.~Wang, J.~F. Kolb, M.~Elkashlan, M.~Chen, M.~Di~Renzo,
  Y.~Hao, J.~Wang, A.~L. Swindlehurst, {\em et~al.}, ``Reconfigurable
  intelligent surfaces for {{6G}} systems: Principles, applications, and
  research directions,'' {\em IEEE Communications Magazine}, vol.~59, no.~6,
  pp.~14--20, 2021.

\bibitem{wu2021intelligent}
Q.~Wu, S.~Zhang, B.~Zheng, C.~You, and R.~Zhang, ``Intelligent reflecting
  surface-aided wireless communications: A tutorial,'' {\em IEEE Transactions
  on Communications}, vol.~69, no.~5, pp.~3313--3351, 2021.

\bibitem{yang2021energy}
Z.~Yang, M.~Chen, W.~Saad, W.~Xu, M.~Shikh-Bahaei, H.~V. Poor, and S.~Cui,
  ``Energy-efficient wireless communications with distributed reconfigurable
  intelligent surfaces,'' {\em IEEE Transactions on Wireless Communications},
  vol.~21, no.~1, pp.~665--679, 2021.

\bibitem{9690635}
X.~Yu, V.~Jamali, D.~Xu, D.~W.~K. Ng, and R.~Schober, ``Smart and
  reconfigurable wireless communications: From {{IRS}} modeling to algorithm
  design,'' {\em IEEE Wireless Communications}, vol.~28, no.~6, pp.~118--125,
  2021.

\bibitem{9838993}
M.~Zhang and X.~Yuan, ``{{IRS}}-aided {{MIMO}} with cascaded {{LoS}} links:
  Channel modelling and full multiplexing region,'' in {\em ICC 2022 - IEEE
  International Conference on Communications}, pp.~2041--2046, 2022.

\bibitem{9827797}
Z.~Esmaeilbeig, K.~V. Mishra, and M.~Soltanalian, ``{{IRS}}-aided radar:
  Enhanced target parameter estimation via intelligent reflecting surfaces,''
  in {\em 2022 IEEE 12th Sensor Array and Multichannel Signal Processing
  Workshop (SAM)}, pp.~286--290, 2022.

\bibitem{ghatak2021placement}
G.~Ghatak, ``On the placement of intelligent surfaces for {{RSSI}}-based
  ranging in mm-wave networks,'' {\em IEEE Communications Letters}, vol.~25,
  no.~6, pp.~2043--2047, 2021.

\bibitem{ntontin2020reconfigurable}
K.~Ntontin, A.-A.~A. Boulogeorgos, D.~Selimis, F.~Lazarakis, A.~Alexiou, and
  S.~Chatzinotas, ``Reconfigurable intelligent surface optimal placement in
  millimeter-wave networks,'' {\em arXiv preprint arXiv:2011.09949}, 2020.

\bibitem{lu2021aerial}
H.~Lu, Y.~Zeng, S.~Jin, and R.~Zhang, ``Aerial intelligent reflecting surface:
  Joint placement and passive beamforming design with {{3D}} beam flattening,''
  {\em IEEE Transactions on Wireless Communications}, vol.~20, no.~7,
  pp.~4128--4143, 2021.

\bibitem{hashida2020intelligent}
H.~Hashida, Y.~Kawamoto, and N.~Kato, ``Intelligent reflecting surface
  placement optimization in air-ground communication networks toward {{6G}},''
  {\em IEEE Wireless Communications}, vol.~27, no.~6, pp.~146--151, 2020.

\bibitem{9712623}
P.-Q. Huang, Y.~Zhou, K.~Wang, and B.-C. Wang, ``Placement optimization for
  multi-{{IRS}}-aided wireless communications: An adaptive differential
  evolution algorithm,'' {\em IEEE Wireless Communications Letters}, vol.~11,
  no.~5, pp.~942--946, 2022.

\bibitem{issa2021using}
M.~Issa and H.~Artail, ``Using reflective intelligent surfaces for indoor
  scenarios: Channel modeling and {{RIS}} placement,'' in {\em 2021 17th
  International Conference on Wireless and Mobile Computing, Networking and
  Communications (WiMob)}, pp.~277--282, IEEE, 2021.

\bibitem{9903366}
G.~Stratidakis, S.~Droulias, and A.~Alexiou, ``Optimal position and orientation
  study of reconfigurable intelligent surfaces in a mobile user environment,''
  {\em IEEE Transactions on Antennas and Propagation}, pp.~1--1, 2022.

\bibitem{zeng2020reconfigurable}
S.~Zeng, H.~Zhang, B.~Di, Z.~Han, and L.~Song, ``Reconfigurable intelligent
  surface ({{RIS}}) assisted wireless coverage extension: {{RIS}} orientation
  and location optimization,'' {\em IEEE Communications Letters}, vol.~25,
  no.~1, pp.~269--273, 2020.

\bibitem{9186137}
E.~Tohidi, R.~Amiri, M.~Coutino, D.~Gesbert, G.~Leus, and A.~Karbasi,
  ``Submodularity in action: From machine learning to signal processing
  applications,'' {\em IEEE Signal Processing Magazine}, vol.~37, no.~5,
  pp.~120--133, 2020.

\bibitem{8537943}
E.~Tohidi, M.~Coutino, S.~P. Chepuri, H.~Behroozi, M.~M. Nayebi, and G.~Leus,
  ``Sparse antenna and pulse placement for colocated {{MIMO}} radar,'' {\em
  IEEE Transactions on Signal Processing}, vol.~67, no.~3, pp.~579--593, 2019.

\bibitem{9387156}
E.~Tohidi, S.~Parsaeefard, M.~A. Maddah-Ali, B.~H. Khalaj, and A.~Leon-Garcia,
  ``Near-optimal robust virtual controller placement in {{5G}} software defined
  networks,'' {\em IEEE Transactions on Network Science and Engineering},
  vol.~8, no.~2, pp.~1687--1697, 2021.

\bibitem{Quad}
{Fraunhofer HHI}, ``Quadriga.'' {https://github.com/fraunhoferhhi/QuaDRiGa},
  2021.

\end{thebibliography}

\end{document}